\documentclass[12pt]{article}

\usepackage[utf8]{inputenc}
\usepackage[T1]{fontenc}
\usepackage{lmodern}
\usepackage{microtype}
\usepackage{xcolor}
\usepackage{rotating}
\usepackage[a4paper, margin=1in]{geometry}
\usepackage{amsmath,amssymb,amsthm}
\usepackage{graphicx}
\usepackage{booktabs}
\usepackage{pdflscape}
\usepackage{caption}
\usepackage{subcaption}
\usepackage{setspace}
\usepackage{lineno}
\usepackage{url}
\usepackage[authoryear]{natbib}
\usepackage[colorlinks=true,linkcolor=blue!50!black,citecolor=blue!50!black,urlcolor=blue!50!black]{hyperref}

\bibpunct{(}{)}{;}{a}{,}{,}
\theoremstyle{definition}
\newtheorem{theorem}{Theorem}
\newtheorem{proposition}{Proposition}

\theoremstyle{definition}
\newtheorem{definition}{Definition}

\DeclareMathOperator{\sgn}{sgn}
\DeclareMathOperator{\median}{median}
\DeclareMathOperator{\sd}{sd}
\DeclareMathOperator{\AMSE}{AMSE}
\DeclareMathOperator*{\argmin}{arg\,min}
\newcommand{\E}{\mathbb{E}}
\newcommand{\R}{\mathbb{R}}

\title{SCOPE: A Family of CDF-Based Wavelet Shrinkage Rules}
\author{Dixon Vimalajeewa$^{1}$, Vijini Lakmini$^{2}$, Malith Premarathna$^{1}$,\\
Fabrizio Ruggeri$^{3}$, and Brani Vidakovic$^{2}$\\[1ex]
{\small $^{1}$Department of Statistics, University of Nebraska, Lincoln, NE 68588, USA}\\
{\small $^{2}$Department of Statistics, Texas A\&M University, College Station, TX 77843, USA}\\
{\small $^{3}$IMATI--CNR, 20133 Milano, Italy}}
\date{\today}

\begin{document}

\maketitle

\begin{abstract}
We introduce SCOPE (Symmetric CDF Oriented Probability Enhanced) shrinkage, a family of wavelet denoising rules based on centered cumulative distribution functions of symmetric unimodal laws. These rules preserve signs, contract coefficient magnitudes, and provide flexible control over shrinkage behavior. We establish their structural properties and introduce neighborhood SCOPE (N-SCOPE), which uses local wavelet energy to determine coefficient attenuation. Two tuning approaches are proposed without requiring the true signal. Fixed SCOPE and N-SCOPE are calibrated using Gaussian white noise for each profile and neighborhood radius, with level-dependent scales that equalize expected noise energy across resolution levels. Adaptive SCOPE and N-SCOPE select attenuation parameters and neighborhood radius for N-SCOPE by minimizing a plug-in version of Stein’s unbiased risk estimate (SURE). Monte Carlo simulations using four Donoho–Johnstone test signals across three sample sizes and three signal-to-noise ratios demonstrate that neighborhood shrinkage improves recovery of locally persistent features but may overshrink isolated coefficients. SURE adaptation substantially reduces this trade-off. With the normal generator, N-SCOPE SURE has about 10\% lower average mean squared error, normalized by universal hard thresholding, than coefficientwise SCOPE SURE and 15-17\% lower than the empirical Bayes posterior median; among the 26 rules compared, only a hidden Markov tree proxy matches or exceeds it.
\end{abstract}


\section{Introduction}
\label{sec:intro}

Wavelet denoising must suppress noise without erasing signal structure.
Small empirical wavelet coefficients are often noise dominated, whereas
large coefficients may encode discontinuities, peaks, or other salient
features. A useful shrinkage rule should therefore attenuate weak
coefficients while retaining substantial ones, with a controlled
transition between the two regimes. Whether neighboring coefficients
should inform this attenuation is a separate consideration: local
information can reinforce genuine features, but excessive pooling
may obscure isolated discontinuities.

The seminal contributions of
\citet{DonohoJohnstone1994,DonohoJohnstone1995} motivated a broad
literature on wavelet shrinkage, encompassing SURE-based selection,
multiple-testing methods \citep{AbramovichBenjamini1996},
cross-validation \citep{Nason1996}, and Bayesian and empirical Bayes
procedures \citep{Vidakovic1998,VidakovicRuggeri2001,ChipmanKolaczykMcCulloch1997,
JohnstoneSilverman2005EBayes}. Smooth and penalized shrinkage rules
\citep{AntoniadisFan2001,FanLi2001,GaoBruce1997,Gao1998,
AttoPastorMercier2008} and neighborhood methods
\citep{CaiSilverman2001Neigh} address complementary aspects of
bias, variance, and local dependence. Our objective is a unified
construction that accommodates distinct attenuation profiles,
incorporates local wavelet information, and permits both prespecified
and data-adaptive tuning.

We introduce \emph{Symmetric CDF Oriented Probability Enhanced} (SCOPE)
shrinkage. For a symmetric unimodal cumulative distribution function
$F$, let $F^*(u)=2F(u)-1$ denote its centered CDF. The coefficientwise
rule is
\begin{equation}
\delta(x;\lambda,k)=x\,|F^*(\lambda x)|^k,
\qquad \lambda>0,\quad k\geq0.
\label{eq:scope-intro}
\end{equation}
The attenuation factor is bounded by one, preserves the sign of each
coefficient, and cannot increase its magnitude. The scale $\lambda$
determines the location of the shrinkage transition, whereas the
exponent $k$ controls its shape. The generating distribution adds
flexibility through the rate at which attenuation approaches one.
Logistic, normal, Laplace, hyperbolic-secant, uniform, and Cauchy
generators yield distinct profiles within the same structural family.
Here $F$ defines a shrinkage profile, not the sampling distribution
of the errors, which is Gaussian in the present analysis.
Through Khintchine's representation of symmetric unimodal laws, we show
that the admissible attenuation profiles are exactly mixtures of
elementary ramps that rise linearly and then saturate.

We further introduce \emph{neighborhood SCOPE} (N-SCOPE). The
attenuation of a detail coefficient depends on the root-mean-square
energy in a moving neighborhood of radius $r$. Ordinary SCOPE is
recovered at $r=0$, whereas $r>0$ allows nearby coefficients to
reinforce a shared local feature, at the risk of overshrinking isolated
coefficients; this trade-off motivates adaptive choice of $r$. Earlier
neighborhood rules use local wavelet energy \citep{CaiSilverman2001Neigh}
and allow SURE selection of the window size
\citep{ZhouCheng2008NeighShrinkSURE}. We derive the Gaussian null risk
of the central coefficient for each CDF profile and radius, which lets
us match expected noise retention across rules before comparing their
recovery of signal.

A central methodological requirement is \emph{non-oracle tuning}.
We consider two complementary strategies. In the \emph{fixed}
procedures, $F$, $k$, and $r$ are prespecified, and level-dependent
scales are determined by matching the risk at a zero coefficient under
Gaussian noise. The resulting calibration equalizes the expected
retained pure-noise energy across shrunk wavelet levels; it does not
optimize risk for the observed signal. In the \emph{SURE-adaptive}
procedures, these scales provide reference values, while a plug-in
Stein criterion selects the exponent and a common scale multiplier
for SCOPE, and also the neighborhood radius for N-SCOPE. Neither
strategy requires the unknown signal to tune a reconstruction.

We assess the methods in three Monte Carlo experiments using the
Doppler, HeaviSine, Bumps, and Blocks signals, three sample sizes,
and three root signal-to-noise ratios. Fixed positive-radius
neighborhoods improve reconstruction of locally persistent features
but can perform poorly at isolated discontinuities. Adaptive radius
selection largely mitigates this limitation: for the normal generator,
N-SCOPE SURE reduces aggregate normalized AMSE relative to
coefficientwise SCOPE SURE by 9.6--10.4\%, depending on sample size.
The procedures are competitive with the reportable implementations
in a 20-method comparator inventory, although performance depends
on signal morphology and several comparators are one-dimensional
adaptations or provisional implementations. We therefore emphasize
comparative evidence within the study design, not uniform dominance
over established denoisers.

Section~\ref{sec:defs} defines SCOPE and N-SCOPE and develops their
properties; Section~\ref{sec:rules} illustrates representative
generators. Section~\ref{sec:calibration} describes Gaussian null-risk
calibration and SURE selection. The three numerical experiments are
reported in Section~\ref{sec:simul}, followed by the discussion in
Section~\ref{sec:conclusions}. The Supplementary Material contains
the penalized-likelihood and MAP interpretation, additional proofs,
and further numerical results.

\section{SCOPE and N-SCOPE: Definitions and Properties}
\label{sec:defs}

Consider the Gaussian nonparametric regression model
\begin{equation}
y_i=f_i+\sigma\varepsilon_i,\qquad
\varepsilon_i\stackrel{\mathrm{iid}}{\sim}N(0,1),\quad i=1,\ldots,n,
\label{eq:scope_np_reg}
\end{equation}
with $\sigma>0$. An orthonormal wavelet transform yields the equivalent
Gaussian sequence model
\begin{equation}
x_i=\theta_i+\sigma\varepsilon_i,\qquad i=1,\ldots,n,
\label{eq:scope_gauss_seq}
\end{equation}
where $x_i$ and $\theta_i$ are the observed and underlying wavelet
coefficients. For theoretical purposes, $\sigma$ may be treated as
known and the coefficients standardized as
\begin{equation}
x_i^*=\theta_i^*+\varepsilon_i,\qquad
x_i^*=x_i/\sigma,\quad \theta_i^*=\theta_i/\sigma.
\label{eq:scope_gauss_seq_std}
\end{equation}
In practice, $\sigma$ is estimated from the data, and the shrunk wavelet
coefficients are mapped back to the signal domain by the inverse
wavelet transform.

\paragraph{Standing assumption.}
Unless otherwise stated, $F$ is a continuous, symmetric, unimodal
distribution function whose centered CDF $F^*(u)=2F(u)-1$ satisfies
$uF^*(u)>0$ for $u\ne0$. Thus $F^*$ is continuous, odd, and
nondecreasing, with $|F^*(u)|\le1$ and $F^*(u)\to1$ as $u\to\infty$.

\begin{definition}[SCOPE shrinkage]
For $\lambda>0$ and $k\ge0$, define
\begin{equation}
\delta(x;\lambda,k)=x\,|F^*(\lambda x)|^k.
\label{eq:scope_def}
\end{equation}
\end{definition}

We adopt the convention $|F^*(\lambda x)|^0\equiv1$, so $k=0$
corresponds to the identity rule. For $k>0$, the attenuation factor
$|F^*(\lambda x)|^k\in[0,1]$ depends on both its scale $\lambda$
and shape $k$; these effects are coupled, as discussed in
Section~\ref{sec:threshold}.

\begin{proposition}[Structural properties of SCOPE]
\label{prop:one}
Under the standing assumption, for every $\lambda>0$ and $k\ge0$, the
rule $\delta(\cdot;\lambda,k)$ is
\begin{enumerate}
\item[(a)] odd;
\item[(b)] sign preserving: $\sgn\delta(x;\lambda,k)=\sgn(x)$ for $x\ne0$;
\item[(c)] contractive in magnitude: $|\delta(x;\lambda,k)|\le|x|$;
\item[(d)] nondecreasing on $\R$;
\item[(e)] continuous on $\R$; and
\item[(f)] relatively tail preserving: $\delta(x;\lambda,k)/x\to1$ as
$|x|\to\infty$.
\end{enumerate}
\end{proposition}

The proof, based on symmetry and monotonicity of $F^*$, is given in
the Supplementary Material. Magnitude contraction does not imply
nonexpansiveness in the Lipschitz sense. Moreover, relative tail
preservation need not entail vanishing absolute shrinkage:
$|x-\delta(x;\lambda,k)|$ can remain bounded away from zero, as
illustrated by the Cauchy generator.

\paragraph{Special cases and limits.}
For odd positive integers $k$, the rule can also be written as
\begin{equation}
\delta(x;\lambda,k)=|x|\{F^*(\lambda x)\}^k.
\label{eq:scope_oddk}
\end{equation}
For the uniform generator, as $k\to\infty$ at fixed $\lambda$,
the rule approaches a hard-threshold-type limit: coefficients with
$|x|<1/\lambda$ vanish, whereas those with
$|x|\ge 1/\lambda$ are retained. For generators with
$|F^*(u)|<1$ at every finite $u$ (including logistic, normal, and
Cauchy), $\delta(x;\lambda,k)\to0$ for every fixed $x$.
Thus increasing $k$ without readjusting $\lambda$ also moves the
effective attenuation threshold.

\subsection{Mixture-of-uniforms representation}

A useful representation connects SCOPE with elementary soft
attenuation profiles.

\begin{theorem}[Mixture-of-uniforms representation]
Let $F$ satisfy the standing assumption. There is a probability measure
$G$ on $(0,\infty)$ such that
\begin{equation}
|F^*(u)|=\int_0^\infty
\min\!\left\{\frac{|u|}{t},1\right\}\,dG(t),
\qquad u\in\R.
\label{eq:mix_absF}
\end{equation}
Consequently,
\begin{equation}
\delta(x;\lambda,k)=x\left[
\int_0^\infty\min\!\left\{\frac{|\lambda x|}{t},1\right\}
\,dG(t)\right]^k.
\label{eq:mix_scope}
\end{equation}
\end{theorem}

\begin{proof}
By the Khintchine representation of symmetric unimodal distributions
\citep{Khintchine1938}, $F$ is a scale mixture of centered uniform
laws,
\begin{equation}
F(x)=\int_0^\infty F_t(x)\,dG(t),
\label{eq:mix_F}
\end{equation}
where $F_t$ is uniform on $[-t,t]$. Since $F$ is continuous, the mixing
law has no atom at zero. For $t>0$,
\begin{equation}
2F_t(x)-1=\sgn(x)\min\!\left\{\frac{|x|}{t},1\right\}.
\label{eq:uniform_centered}
\end{equation}
Integrating and taking absolute values gives \eqref{eq:mix_absF};
substitution into \eqref{eq:scope_def} gives \eqref{eq:mix_scope}.
\end{proof}

The converse also holds: for every probability measure $G$ on
$(0,\infty)$, the right-hand side of \eqref{eq:mix_absF} is the centered
CDF of a symmetric unimodal law satisfying the standing assumption. The
theorem therefore characterizes the attenuation profiles that symmetry
and unimodality allow: after taking the $k$th root, the attenuation
factor on the positive half-line is a mixture of elementary ramps
$\min(u/t,1)$ that rise linearly and then saturate, and hence a concave,
nondecreasing function rising from $0$ to $1$. The representation extends
directly to neighborhood shrinkage, with local energy in place of the
coefficient magnitude.

\subsection{Neighborhood SCOPE: coefficientwise attenuation using local energy}
\label{sec:nscope}

Nearby wavelet coefficients may represent the same local feature.
N-SCOPE incorporates their magnitudes into the attenuation factor while
continuing to estimate each coefficient separately. For a detail
coefficient $d_{j\ell}$ at level $j$ and a radius $r\ge0$, define the
local root-mean-square energy
\begin{equation}
S_{j\ell}(r)=\left\{\frac{1}{2r+1}
\sum_{h=-r}^{r}d_{j,\ell+h}^{\,2}\right\}^{1/2},
\label{eq:nscope-energy}
\end{equation}
with periodic indexing within each level and $2r+1$ not exceeding the
number of coefficients at that level. The N-SCOPE estimator is
\begin{equation}
\widehat\theta_{j\ell}^{\,\mathrm{NS}}
=d_{j\ell}\left[F^*\!\left(\lambda_jS_{j\ell}(r)\right)\right]^k,
\qquad \lambda_j>0,\quad k>0.
\label{eq:nscope-estimator}
\end{equation}
As before, $k=0$ denotes the identity rule. Since $S_{j\ell}(r)\ge0$,
no absolute value is needed around $F^*$. In particular,
$S_{j\ell}(0)=|d_{j\ell}|$, so that $r=0$ recovers coefficientwise
SCOPE. In applications, the calibrated scale $\lambda_j$ may itself
depend on $r$. Boundary neighborhoods are handled by periodic indexing within each resolution level, consistently with the periodized wavelet transform used in our simulations. This preserves a fixed neighborhood size of $2r+1$ for all coefficients. Alternative boundary treatments, including reflection and truncated neighborhoods with adjusted normalization, are also possible.

We emphasize that
N-SCOPE is not a block shrinkage method in the conventional sense: no
common shrinkage factor is applied to a fixed block. Instead, each
coefficient receives its own attenuation, determined by the energy of
the moving, overlapping neighborhood centered at that coefficient.

\begin{proposition}[Structural properties of N-SCOPE]
\label{prop:nscope-structure}
Under the standing assumption, fix $r\ge0$, $\lambda_j>0$, and $k>0$.
At each level $j$, N-SCOPE has the following properties.
\begin{enumerate}
\item[(a)] \emph{Sign equivariance.} Arbitrary coordinatewise sign
reversals change only the corresponding estimated signs; all
attenuation factors remain unchanged.
\item[(b)] \emph{Magnitude and energy contraction.}
\begin{equation}
|\widehat\theta_{j\ell}^{\,\mathrm{NS}}|\le|d_{j\ell}|,
\qquad
\|\widehat{\boldsymbol\theta}_j^{\,\mathrm{NS}}\|_2
\le\|\boldsymbol d_j\|_2.
\label{eq:nscope-contract}
\end{equation}
\item[(c)] \emph{Continuity and own-coordinate monotonicity.} The
mapping is continuous, and with the neighboring coefficients fixed,
$\widehat\theta_{j\ell}^{\,\mathrm{NS}}$ is nondecreasing in $d_{j\ell}$.
\item[(d)] \emph{Relative tail preservation.} With all other
coefficients fixed,
\begin{equation}
\frac{\widehat\theta_{j\ell}^{\,\mathrm{NS}}}{d_{j\ell}}
\longrightarrow1\qquad\text{as }|d_{j\ell}|\to\infty.
\label{eq:nscope-tail}
\end{equation}
\end{enumerate}
\end{proposition}

\begin{proof}
Write $m=2r+1$ and
$w_{j\ell}=[F^*(\lambda_jS_{j\ell}(r))]^k$.
By the standing assumption, $0\le w_{j\ell}\le1$,
and $w_{j\ell}$ depends on the coefficients only
through their squares. Thus coordinatewise sign
reversals leave all attenuation factors unchanged,
proving (a). Moreover,
$|\widehat\theta_{j\ell}^{\,\mathrm{NS}}|
\le |d_{j\ell}|$; summing the squared inequalities
over $\ell$ establishes (b).

Continuity follows from that of $F^*$ and
$S_{j\ell}(r)$. To establish (c), fix the neighboring
coefficients and write
$$
A=\sum_{\substack{h=-r\\h\ne0}}^r
d_{j,\ell+h}^{\,2}\ge0,
\qquad
q(t)=\left[
F^*\!\left(\lambda_j
\sqrt{\frac{t^2+A}{m}}\right)
\right]^k,\quad t\ge0.
$$
The function $q$ is nonnegative and nondecreasing.
Consequently, $xq(|x|)$ is nondecreasing on
$[0,\infty)$ and on $(-\infty,0]$, and the
two branches meet continuously at zero.
This proves (c).

Finally,
$$
S_{j\ell}(r)\ge
\frac{|d_{j\ell}|}{\sqrt m}.
$$
Since $F^*$ is nondecreasing on $[0,\infty)$
and $F^*(u)\to1$ as $u\to\infty$,
$$
1\ge
\frac{\widehat\theta_{j\ell}^{\,\mathrm{NS}}}
{d_{j\ell}}
\ge
\left[
F^*\!\left(
\frac{\lambda_j|d_{j\ell}|}{\sqrt m}
\right)
\right]^k
\longrightarrow1.
$$
The squeeze theorem establishes (d).
\end{proof}

The dependence on neighboring magnitudes can be stated more precisely.

\begin{proposition}[Neighborhood reinforcement and local attenuation]
\label{prop:nscope-reinforcement}
Suppose, in addition, that $F^*$ is differentiable on $(0,\infty)$ with
$(F^*)'\ge0$. For $S_{j\ell}(r)>0$ and $|h|\le r$,
\begin{equation}
\frac{\partial w_{j\ell}}{\partial |d_{j,\ell+h}|}
=\frac{k\lambda_j|d_{j,\ell+h}|}{(2r+1)S_{j\ell}(r)}
\left[F^*\!\left(\lambda_jS_{j\ell}(r)\right)\right]^{k-1}
(F^*)'\!\left(\lambda_jS_{j\ell}(r)\right)\ge0.
\label{eq:nscope-reinforcement}
\end{equation}
The monotonicity conclusion holds without differentiability. If
$F^*(u)=c_Fu+o(u)$ as $u\downarrow0$ with $c_F>0$, and all coefficients
in a neighborhood approach zero proportionally at rate $|t|$, then
$\widehat\theta_{j\ell}^{\,\mathrm{NS}}=O(|t|^{k+1})$. If instead
$d_{j\ell}\to0$ while at least one nonzero neighbor is held fixed, then
\begin{align}
\widehat\theta_{j\ell}^{\,\mathrm{NS}}
&=d_{j\ell}\left[F^*(\lambda_jS_0)\right]^k+o(d_{j\ell}),
\nonumber\\
S_0&=\left\{\frac{1}{2r+1}
\sum_{\substack{h=-r\\h\ne0}}^{r}d_{j,\ell+h}^{\,2}
\right\}^{1/2}>0.
\label{eq:nscope-isolated}
\end{align}
\end{proposition}

\begin{proof}
The derivative follows from the chain rule and the derivative of
$S_{j\ell}(r)$; monotonicity without differentiability follows by
composition. The two local expansions follow by substituting
$d_{j,\ell+h}=ta_h$ and $S_{j\ell}(r)\to S_0$, respectively, into
\eqref{eq:nscope-estimator}.
\end{proof}

\paragraph{The isolated-coefficient trade-off.}
When all neighboring coefficients vanish,
$S_{j\ell}(r)=|d_{j\ell}|/\sqrt{2r+1}$, and the estimator reduces to
\begin{equation}
\widehat\theta_{j\ell}^{\,\mathrm{NS}}
=\delta\!\left(d_{j\ell};\frac{\lambda_j}{\sqrt{2r+1}},k\right).
\label{eq:nscope-spike}
\end{equation}
At a common nominal scale, a larger neighborhood consequently
attenuates an isolated coefficient more strongly, even though
neighboring signal energy may protect coefficients associated with
a shared feature. Equation~\eqref{eq:nscope-spike} holds at a fixed
$\lambda_j$; in the procedures of Section~\ref{sec:calibration},
scales are recalibrated separately for each radius. Experiment~I
(Section~\ref{sec:simul-fixed}) examines this trade-off empirically.

\paragraph{Mixture representation.}
With the mixing law $G$ of the preceding theorem, N-SCOPE admits the
representation
\begin{equation}
\widehat\theta_{j\ell}^{\,\mathrm{NS}}
=d_{j\ell}\left[\int_0^\infty
\min\!\left\{\frac{\lambda_jS_{j\ell}(r)}{t},1\right\}
\,dG(t)\right]^k.
\label{eq:nscope-mixture}
\end{equation}
The structural properties above do not assert global monotonicity or
nonexpansiveness of the multivariate N-SCOPE mapping.

\paragraph{Penalized-likelihood and MAP interpretation.}
Under appropriate regularity conditions, coefficientwise SCOPE
($r=0$) has an exact penalized Gaussian least-squares representation
with an even penalty that is nondecreasing in coefficient magnitude.
This construction yields a generalized maximum a posteriori (MAP)
interpretation; whether the associated prior is proper depends on
the tail behavior of the shrinkage rule. The scalar argument does not
extend automatically to overlapping neighborhoods ($r>0$).
Conditions, examples, and proofs appear in Section~S1 of the
Supplementary Material.

\section{Representative SCOPE Rules}
\label{sec:rules}

We examine three complementary generators: logistic, uniform, and
Cauchy. They exhibit, respectively, rapid and smooth saturation of
the centered CDF, exact saturation beyond a finite argument, and
slow tail saturation. These examples make explicit how the generating
distribution influences attenuation near zero and the retention of
large coefficients.

\subsection{The logistic SCOPE rule}

For the standard logistic distribution, $F(u)=(1+e^{-u})^{-1}$ and
$F^*(u)=2F(u)-1=\tanh(u/2)$. Absorbing the factor $1/2$ into the scale
parameter gives
\begin{equation}
\delta(x;\lambda,k)=x\,|\tanh(\lambda x)|^k,
\qquad \lambda>0,\quad k\ge0.
\label{eq:scope_logistic}
\end{equation}
For fixed $\lambda$ and $k$, its behavior near zero is
\begin{equation}
\delta(x;\lambda,k)=x\,|\lambda x|^k\{1+O(x^2)\},
\qquad x\to0.
\end{equation}
Thus, the rule attenuates small coefficients at order $k+1$.
Because $|\tanh(\lambda x)|$ approaches one exponentially fast,
$\delta(x;\lambda,k)/x\to1$ as $|x|\to\infty$. Increasing $k$
at fixed $\lambda$ sharpens attenuation but also shifts the
effective transition outward (Figure~\ref{fig:SCOPETanh}).

\begin{figure}[htbp]
\centering
\includegraphics[width=0.80\textwidth]{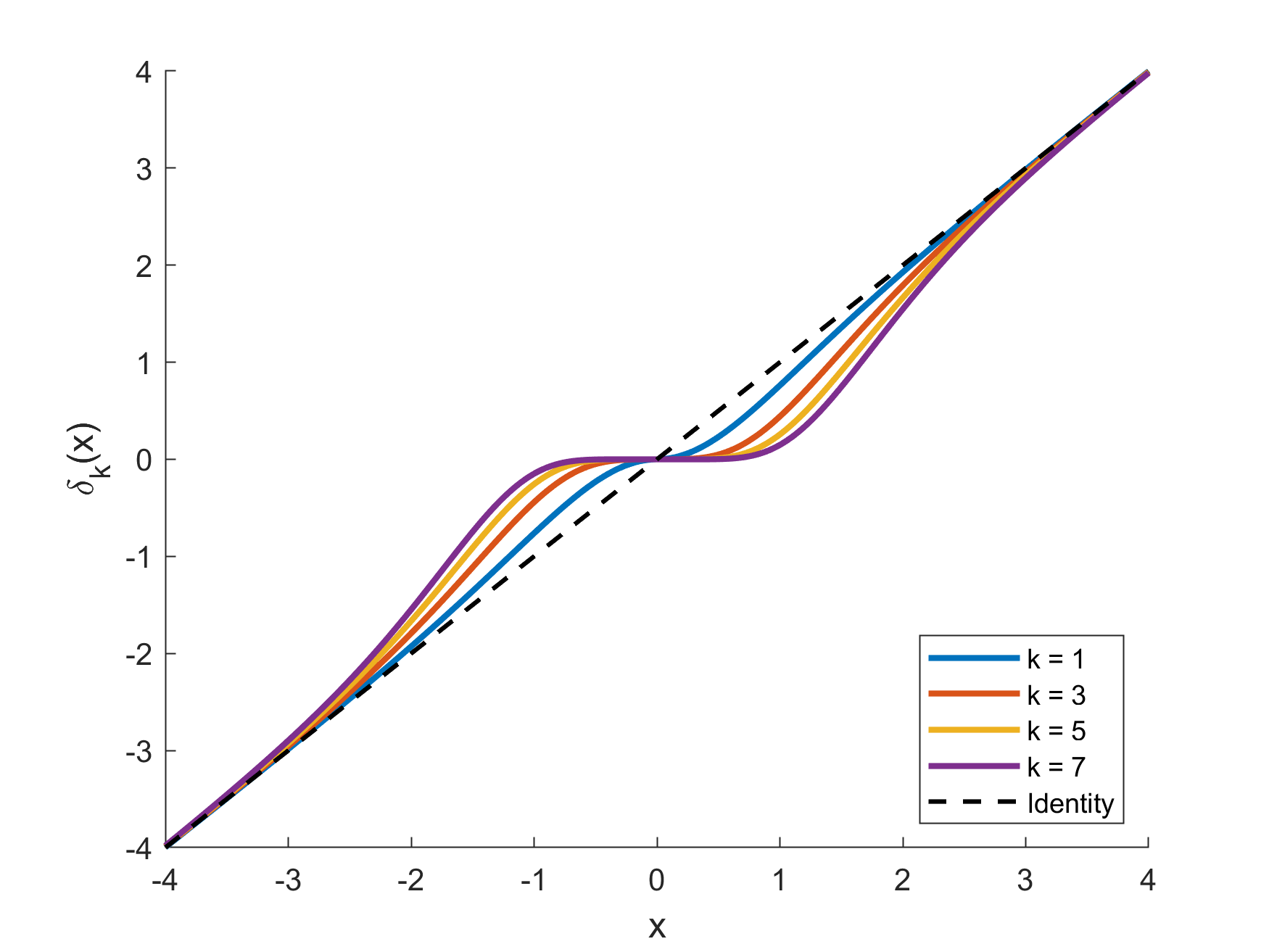}
\caption{Logistic SCOPE rules for $\lambda=1$ and $k=1,3,5,7$, with the
identity line for comparison.}
\label{fig:SCOPETanh}
\end{figure}

\subsection{The uniform SCOPE rule}

For $F$ uniform on $[-1,1]$, the centered CDF is
$F^*(u)=\sgn(u)\min\{|u|,1\}$. Consequently,
\begin{equation}
\delta(x;\lambda,k)=x\,\min\{|\lambda x|,1\}^{k},
\qquad \lambda>0,\quad k\ge0,
\label{eq:scope_uniform}
\end{equation}
or, equivalently,
\begin{equation}
\delta(x;\lambda,k)=
\begin{cases}
x\,|\lambda x|^k, & |\lambda x|\le1,\\[0.5ex]
x, & |\lambda x|>1,
\end{cases}
\label{eq:scope_uniform_piecewise}
\end{equation}
As elsewhere, $k=0$ gives the identity. Unlike logistic SCOPE,
the uniform rule reproduces coefficients exactly for
$|x|>1/\lambda$. It is continuous and nondecreasing, with
polynomial attenuation of order $k+1$ near zero; for $k>0$, however,
its slope changes discontinuously at $|x|=1/\lambda$
(Figure~\ref{fig:SCOPEU}).

\begin{figure}[htbp]
\centering
\includegraphics[width=0.70\textwidth]{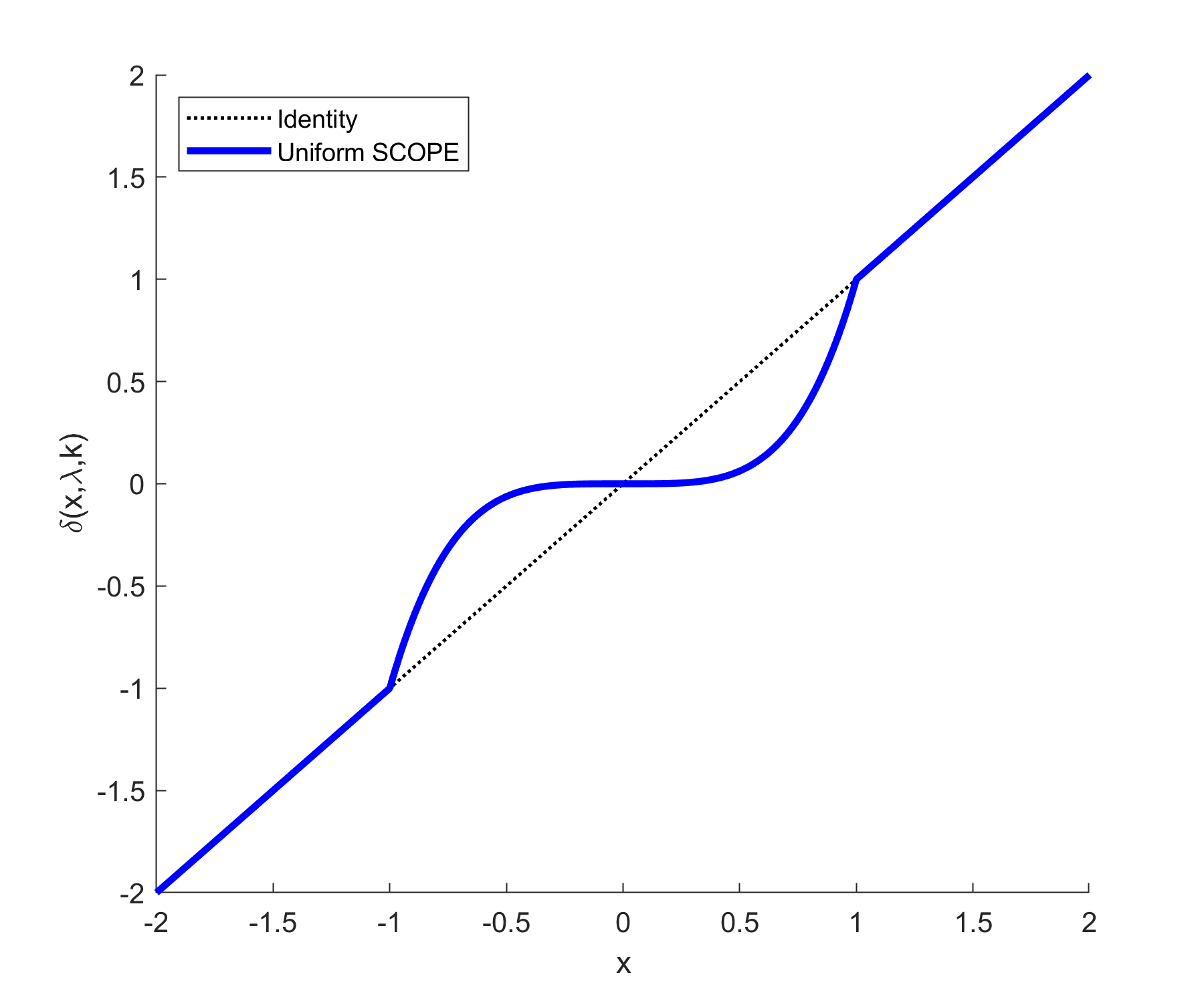}
\caption{Uniform SCOPE rule for $\lambda=1$ and $k=3$: polynomial
attenuation followed by the exact identity.}
\label{fig:SCOPEU}
\end{figure}

\subsection{The Cauchy SCOPE rule}

For the standard Cauchy distribution, $F(u)=1/2+\pi^{-1}\arctan(u)$,
hence $F^*(u)=(2/\pi)\arctan(u)$, and the rule is
\begin{equation}
\delta(x;\lambda,k)
=x\left|\frac{2}{\pi}\arctan(\lambda x)\right|^k,
\qquad \lambda>0,\quad k\ge0.
\label{eq:scope_cauchy}
\end{equation}
Near zero,
\begin{equation}
\delta(x;\lambda,k)
=x\left|\frac{2\lambda x}{\pi}\right|^k\{1+O(x^2)\},
\qquad x\to0,
\end{equation}
so the local attenuation again has order $k+1$. The distinguishing
feature of the Cauchy rule is its slow saturation in the tails: for
fixed $k$ and $\lambda$,
\begin{equation}
\delta(x;\lambda,k)
=x-\frac{2k}{\pi\lambda}\sgn(x)
+O(|x|^{-1}),\qquad |x|\to\infty.
\end{equation}
Relative attenuation vanishes in the tails, but absolute shrinkage
$|x-\delta(x;\lambda,k)|$ converges to $2k/(\pi\lambda)$.
Consequently, the Cauchy generator continues to attenuate even
moderately large coefficients and may incur additional bias
(Figure~\ref{fig:SCOPE3}).

\begin{figure}[htbp]
\centering
\includegraphics[width=0.70\textwidth]{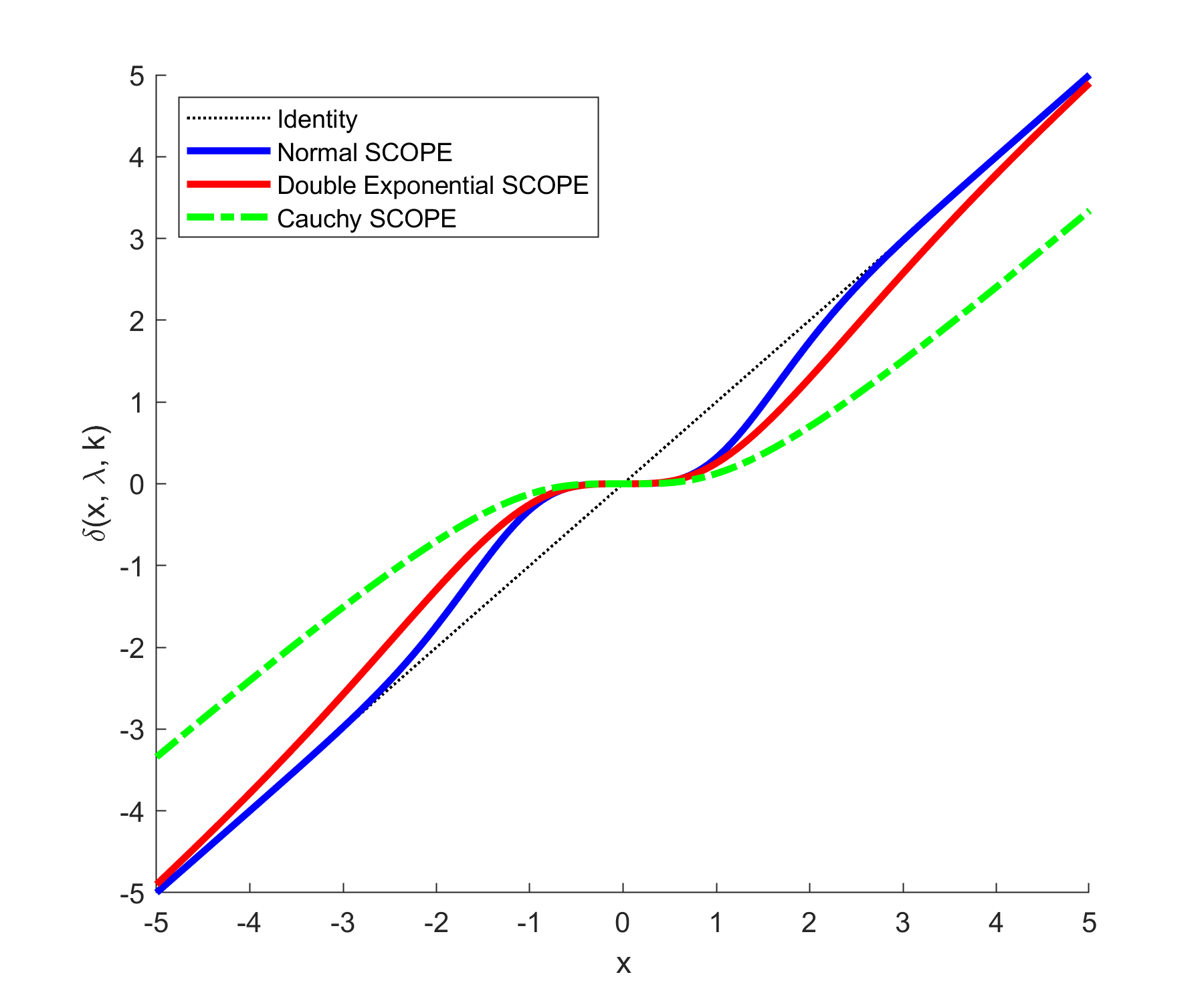}
\caption{Normal, double-exponential (Laplace), and Cauchy SCOPE rules
for $\lambda=1$ and $k=3$. The Cauchy generator exhibits more persistent
attenuation in the tails.}
\label{fig:SCOPE3}
\end{figure}

\subsection{Other generating distributions}
\label{sec:other-cdfs}

The numerical experiments also include normal, Laplace, and
hyperbolic-secant generators. Table~\ref{tab:cdfs} gives all six
centered CDFs under the adopted scale conventions, their slopes at
the origin, and their positive-tail behavior. In particular, the
logistic definition absorbs the factor $1/2$ appearing in the
standard logistic centered CDF.

\begin{table}[t]
\centering\small
\caption{Centered CDF generators, slopes at the origin, and
positive-tail behavior as $u\to\infty$. Here $\Phi$ and $\phi$ denote
the standard normal CDF and density.}
\label{tab:cdfs}
\begin{tabular}{llll}
\toprule
Generator & $F^*(u)$ & $(F^*)'(0)$ & $1-F^*(u)$ \\
\midrule
Logistic & $\tanh(u)$ & $1$ & $\sim 2e^{-2u}$ \\
Normal & $2\Phi(u)-1$ & $\sqrt{2/\pi}$ & $\sim 2\phi(u)/u$ \\
Laplace & $\sgn(u)(1-e^{-|u|})$ & $1$ & $=e^{-u}$ \\
Hyperbolic secant & $(2/\pi)\arctan\{\sinh(\pi u/2)\}$ & $1$ & $\sim(4/\pi)e^{-\pi u/2}$ \\
Uniform & $\max\{-1,\min(u,1)\}$ & $1$ & $=0$ for $u\ge1$ \\
Cauchy & $(2/\pi)\arctan(u)$ & $2/\pi$ & $\sim2/(\pi u)$ \\
\bottomrule
\end{tabular}
\end{table}

The Laplace centered CDF is continuously differentiable but not
twice differentiable at zero; the uniform centered CDF is
nondifferentiable at $u=\pm1$. Neither irregularity affects the
Gaussian null-risk calibration of Section~\ref{sec:calibration},
which requires continuity and monotonicity rather than global
smoothness.

Many other symmetric distributions could serve as SCOPE generators, including the raised-cosine, symmetric double-Pareto, and Tukey biweight distribution, for example. More generally, any distribution supported on the positive half-line can be symmetrized to yield a distribution on $\mathbb{R}$. Exploring these broader classes of generators and their effects on shrinkage behavior and risk properties is an interesting direction for future research.
\section{Non-oracle Calibration and SURE Selection}
\label{sec:calibration}

We calibrate SCOPE and N-SCOPE without reference to the unknown
signal. The fixed procedures prespecify $F$, $k$, and $r$, then
determine level-dependent scales from a common Gaussian null-risk
criterion. The adaptive procedures use those scales as reference
values and select the attenuation parameters, the neighborhood
radius where applicable, and a common scale multiplier by minimizing
a plug-in version of Stein's unbiased risk estimate (SURE). Neither procedure
minimizes reconstruction error against the true signal.

\subsection{Effective attenuation threshold}
\label{sec:threshold}

For $k>0$ and an attenuation fraction $\gamma\in(0,1)$, define
\begin{equation}
q_F(k;\gamma)
  =(F^*)^{-1}(\gamma^{1/k})
  =F^{-1}\!\left(\frac{1+\gamma^{1/k}}{2}\right).
\label{eq:qF}
\end{equation}
The attenuation of a coefficient equals $\gamma$ at magnitude
$q_F(k;\gamma)/\lambda$. In particular,
$t=q_F(k;1/2)/\lambda$ is the \emph{half-attenuation point},
which provides an interpretable effective threshold. The
parameterization $(t,k)$ distinguishes transition location from
shape, unlike the coupled parameters $(\lambda,k)$. Since
$1-2^{-1/k}\sim(\log2)/k$, the tail of $F$ governs how $\lambda$
must change to maintain a fixed $t$ as $k$ varies. Specifically,
$q_F(k;1/2)$ is bounded for the uniform generator and grows
asymptotically as $\sqrt{2\log k}$ for the normal,
$\log(k/\log2)$ for the Laplace, and $2k/(\pi\log2)$ for
the Cauchy generator.

\subsection{Gaussian null-risk calibration for SCOPE and N-SCOPE}
\label{sec:nullrisk}

Let $d_{j\ell}$ be a wavelet detail coefficient. N-SCOPE estimates its
signal component from the local root-mean-square energy $S_{j\ell}(r)$
of $2r+1$ coefficients \eqref{eq:nscope-energy}:
\begin{equation}
\widehat\theta_{j\ell}(F,k,r,\lambda_j)
 = d_{j\ell}\left[F^*\!\left(\lambda_j S_{j\ell}(r)\right)\right]^k,
\qquad
S_{j\ell}(r)=\left\{\frac{1}{2r+1}
   \sum_{h=-r}^{r}d_{j,\ell+h}^2\right\}^{1/2}.
\label{eq:cal-nscope}
\end{equation}
Setting $r=0$ gives $S_{j\ell}(0)=|d_{j\ell}|$ and recovers
coefficientwise SCOPE. For $r>0$, neighboring coefficients modify
the attenuation of the \emph{central} coefficient; the method does
not apply a common factor to a fixed block. The null distribution
of local energy depends on $r$, so different radii require distinct
calibration constants, even for the same generator and exponent.

\begin{proposition}[Null risk of N-SCOPE]
\label{prop:nullrisk}
Fix $k>0$ and $r\geq0$, and assume $2r+1\le n_j$, so that the window
contains $2r+1$ distinct coefficients. Under the Gaussian sequence model
\eqref{eq:scope_gauss_seq}, suppose that $\theta_{j,\ell+h}=0$ for
$|h|\leq r$, and put $a=\lambda_j\sigma$. The standardized risk at the
central null coefficient is
\begin{equation}
\frac{\E_0[(\widehat\theta_{j\ell}^{\,\mathrm{NS}})^2]}{\sigma^2}
 = R_{F,r,k}(a)
 :=\E\!\left[Z_0^2\{F^*(aS)\}^{2k}\right],
 \qquad S^2=\frac{Z_0^2+V}{2r+1},
\label{eq:nullrisk}
\end{equation}
where $Z_0\sim N(0,1)$ and $V\sim\chi^2_{2r}$ are independent, with
$V\equiv0$ for $r=0$. Equivalently, for $U\sim\chi^2_{2r+1}$,
\begin{equation}
R_{F,r,k}(a)
 = \E\!\left[\frac{U}{2r+1}
 \left\{F^*\!\left(a\sqrt{\frac{U}{2r+1}}\right)\right\}^{2k}\right].
\label{eq:nullrisk-chi}
\end{equation}
Under the standing assumption, $R_{F,r,k}(a)$ is continuous and
strictly increasing from $0$ to $1$ on $a\in(0,\infty)$. Consequently,
for each $\eta\in(0,1)$ there is a unique $a_{F,r,k}(\eta)>0$
satisfying $R_{F,r,k}\{a_{F,r,k}(\eta)\}=\eta$. Rescaling the generator
from $F(u)$ to $F(u/s)$ multiplies $a_{F,r,k}(\eta)$ by $s$ and leaves
the calibrated shrinkage rule unchanged.
\end{proposition}

\begin{proof}
For a null neighborhood, write $d_{j,\ell+h}=\sigma Z_h$ with
independent standard normal $Z_h$. Then
$\{S_{j\ell}(r)/\sigma\}^2=(Z_0^2+V)/(2r+1)$, which gives
\eqref{eq:nullrisk}. Conditioning on $U=Z_0^2+V$ and using spherical
symmetry yields $\E(Z_0^2\mid U)=U/(2r+1)$ and hence
\eqref{eq:nullrisk-chi}. The stated properties follow from
monotonicity, dominated convergence, and the positive density of $S$ on
$(0,\infty)$. The scale claim follows from $F_s^*(u)=F^*(u/s)$.
\end{proof}

\paragraph{Common null-risk target.}
We anchor the finest-level calibration to the universal threshold
$\sigma\sqrt{2\log n}$ \citep{DonohoJohnstone1994}. For normal
SCOPE with $k=5$, define
\begin{equation}
a_n=\frac{\Phi^{-1}\{(1+2^{-1/5})/2\}}{\sqrt{2\log n}},
\qquad
\eta(n)=\E\left[Z^2\{2\Phi(a_n|Z|)-1\}^{10}\right],
\qquad Z\sim N(0,1),
\label{eq:eta-n}
\end{equation}
where $\eta(n)$ is the standardized null risk of the normal rule
whose half-attenuation point coincides with
$\sigma\sqrt{2\log n}$. Numerically,
$\eta(512)=0.01370$, $\eta(1024)=0.01003$, and
$\eta(2048)=0.00748$. We use this same $\eta(n)$ for every
generator and radius, thereby matching pure-noise suppression, not
risk at nonzero coefficients. Table~\ref{tab:aF} lists the constants
for $k=5$ at $\eta=0.01\approx\eta(1024)$.

\begin{table}[t]
\centering\small
\caption{Gaussian null-risk calibration constants $a_{F,r,5}(0.01)$ for
coefficientwise SCOPE ($r=0$) and N-SCOPE ($r=1,2,3$), computed from
\eqref{eq:nullrisk-chi}. At the finest level,
$\lambda_{J-1}=a_{F,r,5}(0.01)/\widehat\sigma$ for the fixed procedure.}
\label{tab:aF}
\begin{tabular}{lcccc}
\toprule
Generator & $r=0$ & $r=1$ & $r=2$ & $r=3$ \\
\midrule
Logistic          & 0.3443 & 0.4932 & 0.5571 & 0.5941 \\
Normal            & 0.4071 & 0.5865 & 0.6646 & 0.7103 \\
Laplace           & 0.4826 & 0.6840 & 0.7678 & 0.8153 \\
Hyperbolic secant & 0.3627 & 0.5175 & 0.5832 & 0.6209 \\
Uniform           & 0.2545 & 0.3766 & 0.4341 & 0.4693 \\
Cauchy            & 0.7936 & 1.1018 & 1.2232 & 1.2900 \\
\bottomrule
\end{tabular}
\end{table}

For each generator, $a_{F,r,5}(0.01)$ increases with the radius:
averaging additional pure-noise coefficients concentrates the local
energy, requiring a larger nominal scale to retain the specified
null risk. Applying the coefficientwise calibration constant at
$r>0$ would therefore \emph{not} preserve the common null-risk target.

\subsection{Level schedule: equal retained noise energy}
\label{sec:enb}

The finest levels typically contain numerous noise-dominated
coefficients, whereas coarser levels contain fewer coefficients and
may carry more signal energy. We therefore prescribe equal total
expected retained pure-noise energy across the shrunk levels.
With $n_j=2^j$ coefficients at level $j$, the targets satisfy
\begin{equation}
n_j\eta_j=n_{J-1}\eta,
\qquad
\eta_j=\eta\,2^{J-1-j},\quad j=j_0,\ldots,J-1,
\label{eq:enb}
\end{equation}
provided that $\eta\,2^{J-1-j_0}<1$. For a requested starting level,
we increase $j_0$ until $\eta(N)2^{J-1-j_0}<1$, leaving all coarser
coefficients unchanged; with $j_0=3$ requested, this gives $j_0=3$ for
$N\le2048$ and $j_0=4$ for $N=4096$. The fixed calibration for either
SCOPE or N-SCOPE is
\begin{equation}
\lambda_{j,r}(F,k)=\frac{a_{F,r,k}(\eta_j)}{\widehat\sigma}.
\label{eq:cal-fixed-scale}
\end{equation}
For logistic SCOPE with $k=5$, the calibrated constants range
from $1.38$ at the coarsest shrunk level to $0.344$ at the
finest; thus fine-scale detail coefficients undergo stronger
attenuation. N-SCOPE uses the same level schedule but different,
radius-dependent constants.

\paragraph{Rationale.}
The calibration defines a common noise-based reference across
generators and neighborhood radii; it is not an estimate of
signal-dependent risk-optimal scales. Anchoring $\eta(n)$ to the
universal threshold accommodates sample size, while
$\eta_j(n)=\eta(n)2^{J-1-j}$ equalizes retained expected noise
energy across shrunk levels. The coarsest shrunk level $j_0$ is
chosen to ensure $\eta_j(n)<1$ throughout. These choices define
a reproducible, non-oracle procedure and provide reference scales
for SURE selection, without claiming optimality at nonzero
coefficients.

\subsection{Fixed SCOPE and fixed N-SCOPE}
\label{sec:procedure}

A fixed procedure specifies $(F,k,r,\eta,j_0)$ before observing the
data, computes $a_{F,r,k}(\eta_j)$ using
\eqref{eq:nullrisk-chi}--\eqref{eq:enb}, and estimates the noise
standard deviation from the finest-level details by
\begin{equation}
\widehat\sigma=\frac{\median_{\ell}|d_{J-1,\ell}|}{0.6745}.
\label{eq:cal-sigmahat}
\end{equation}
It substitutes \eqref{eq:cal-fixed-scale} into
\eqref{eq:cal-nscope}, leaves scaling coefficients unchanged, and
applies the inverse wavelet transform. Fixed SCOPE uses $r=0$;
fixed N-SCOPE uses a prespecified $r>0$. Although scales depend
on resolution, no signal-specific risk optimization is involved.
Replacing $\sigma$ with $\widehat\sigma$ makes null-risk
matching approximate rather than exact.

\subsection{SURE-adaptive SCOPE and N-SCOPE}
\label{sec:sure-calibration}

To adapt the attenuation profile and neighborhood radius to the
observed signal, we introduce a common multiplier $b>0$ for the
noise-calibrated scales. For fixed generator $F$, candidates are
\begin{align}
\widehat\theta_{j\ell}(k,r,b)
 &= d_{j\ell}\left[F^*\!\left(
  \lambda_{j,r}(F,k,b)S_{j\ell}(r)\right)\right]^k,\nonumber\\
\lambda_{j,r}(F,k,b)
 &= \frac{b\,a_{F,r,k}(\eta_j)}{\widehat\sigma}.
\label{eq:simul-sure-scale}
\end{align}
The reference value $b=1$ gives the fixed null-risk scale;
$b<1$ strengthens attenuation, whereas $b>1$ weakens it.
With known $\sigma$, the prescribed null-risk target is attained
at $b=1$; at other values, the calibration provides a reference
rather than a constraint.

The numerical study uses the prespecified grids
\begin{equation}
k\in\{3,5,7\},\qquad r\in\{0,1,2,3\},\qquad
b\in\mathcal B=\{0.65,0.80,1.00,1.20,1.50\}.
\label{eq:simul-sure-grids}
\end{equation}
SCOPE SURE fixes $r=0$ and selects $(k,b)$ from 15 candidates;
N-SCOPE SURE selects $(k,r,b)$ from 60 candidates. Selected parameters
are common to all shrunk detail levels in a realization, although
the local energies $S_{j\ell}(r)$ vary across coefficient locations.
The generator is prespecified for each procedure and is not selected
by SURE from the data; we instead compare generator families in the
simulation study.

For a fixed candidate with known $\sigma$, Stein's identity provides
an unbiased estimate of squared-error risk, apart from the unchanged
contribution of the unshrunk coefficients. Substituting
$\widehat\sigma$ gives the plug-in criterion
\begin{equation}
\widehat{\mathcal R}_{\mathrm{SURE}}(k,r,b)
 = \sum_{j=j_0}^{J-1}\sum_{\ell}
 \left\{\bigl[\widehat\theta_{j\ell}(k,r,b)-d_{j\ell}\bigr]^2
 +2\widehat\sigma^2
 \frac{\partial\widehat\theta_{j\ell}(k,r,b)}{\partial d_{j\ell}}
 -\widehat\sigma^2\right\}.
\label{eq:simul-nscope-sure}
\end{equation}
Put $m=2r+1$, $S=S_{j\ell}(r)$, and $u=\lambda_{j,r}(F,k,b)S$. For $S>0$
and $2r+1\le n_j$, so that $d_{j\ell}$ enters its own window once,
the required diagonal derivative is
\begin{equation}
\frac{\partial\widehat\theta_{j\ell}}{\partial d_{j\ell}}
 = [F^*(u)]^k
 +\frac{k\lambda_{j,r}d_{j\ell}^{2}}{mS}
  [F^*(u)]^{k-1}(F^*)'(u).
\label{eq:simul-nscope-divergence}
\end{equation}
For coefficientwise SCOPE ($r=0$) this simplifies to
\begin{equation}
\frac{\partial\widehat\theta_{j\ell}}{\partial d_{j\ell}}
 = [F^*(\lambda_{j,0}|d_{j\ell}|)]^k
 +k\lambda_{j,0}|d_{j\ell}|
  [F^*(\lambda_{j,0}|d_{j\ell}|)]^{k-1}
  (F^*)'(\lambda_{j,0}|d_{j\ell}|).
\label{eq:cal-scope-divergence}
\end{equation}
At $S=0$, derivatives are defined by continuity; the uniform
generator is differentiable almost everywhere. Overlapping
neighborhoods produce off-diagonal Jacobian terms, but Stein's
divergence involves only the diagonal. When evaluating these terms,
$\widehat\sigma$ and the candidate scales are held fixed. The
plug-in criterion is consequently not exactly unbiased for the
fully data-dependent procedure after noise estimation and parameter
selection.

The selected tuning parameters are
\begin{align}
(\widehat k,\widehat b)_{\mathrm{SCOPE}}
 &= \argmin_{k\in\{3,5,7\},\,b\in\mathcal B}
     \widehat{\mathcal R}_{\mathrm{SURE}}(k,0,b),\nonumber\\
(\widehat k,\widehat r,\widehat b)_{\text{N-SCOPE}}
 &= \argmin_{k\in\{3,5,7\},\,r\in\{0,1,2,3\},\,b\in\mathcal B}
     \widehat{\mathcal R}_{\mathrm{SURE}}(k,r,b).
\label{eq:cal-sure-select}
\end{align}
Thus, coefficientwise SCOPE SURE is nested within N-SCOPE SURE at
$r=0$. Enlarging the candidate set cannot increase the minimum
observed SURE, but it need not reduce the realized squared-error
risk. Experiment~II (Section~\ref{sec:simul-sure}) evaluates the
adaptive procedures against their fixed references using independent
noise realizations.

\section{Simulation Study}
\label{sec:simul}

The experiments address three questions: whether neighborhood energy
improves noise-calibrated SCOPE (Experiment~I), whether SURE identifies
useful attenuation profiles and neighborhoods without oracle
information (Experiment~II), and how the prespecified SCOPE procedures
compare with other denoisers on a separate Monte Carlo experiment
(Experiment~III). The design spans three sample sizes and three
signal strengths. We present detailed comparisons at $N=1024$ and
root SNR $5$, together with graphical sensitivity analyses and
normalized-risk summaries across conditions. Complete numerical
results and implementation details appear in Section~S5 of the
Supplementary Material.

\subsection{Design, noise calibration, and performance criteria}
\label{sec:simul-design}

We sample the four Donoho--Johnstone test functions---Doppler,
HeaviSine, Bumps, and Blocks---at $t_i=i/N$, with
$N\in\{512,1024,2048\}$. Observations are
$y_i=f_N(t_i)+\varepsilon_i$, where
$\varepsilon_i\stackrel{\mathrm{iid}}{\sim}N(0,1)$.
We define signal strength by the \emph{root} signal-to-noise ratio
\begin{equation}
\rho_{\rm root} = \frac{\sd_N\{f_N(t_i)\}}{\sigma}
\in\{3,5,7\}, \qquad \sigma=1,
\label{eq:simul-root-snr}
\end{equation}
The target is obtained by multiplying the original test signal
$f_0$ by $\rho_{\rm root}/\sd_N(f_0)$, with
$\sd_N(f_0)=\{N^{-1}\sum_i(f_0(t_i)-\bar f_0)^2\}^{1/2}$.
Noise samples are not individually rescaled. Thus, $3,5,7$ denote
ratios of standard deviations, not the conventional power ratios
$N^{-1}\sum_i f_N(t_i)^2/\sigma^2$; this distinction is relevant
because the test signals need not have mean zero. All reported
results use this root-SNR convention.

All procedures use periodized orthonormal wavelet transforms: Haar
for Blocks, Daubechies-3 (six taps) for Bumps, and Symmlet-4
(eight taps) for Doppler and HeaviSine. With $N=2^J$, the eight
scaling coefficients at level $j_0=3$ are left unchanged, and
details at levels $j=3,\ldots,J-1$ are shrunk. For each
realization, $\sigma$ is estimated from the finest detail level
using \eqref{eq:cal-sigmahat}, with the more precise normalizing
constant $0.67448975$. The reference values are
$\eta(512)=0.01370$, $\eta(1024)=0.01003$, and
$\eta(2048)=0.00748$, with
$\eta_j(N)=\eta(N)2^{J-1-j}$ as in
Section~\ref{sec:calibration}. In particular,
$\eta_3(2048)\simeq0.957$; this permissible but relatively weak
calibration at the coarsest shrunk level follows from the
prespecified equal-energy schedule rather than optimization.

Each of the $4\times3\times3=36$ signal--strength--sample-size
combinations uses $M=1000$ Monte Carlo replications per experiment.
Within a condition, procedures are evaluated on paired Gaussian
noise realizations; Experiment~III uses random streams independent
of Experiments~I and~II. Let $L_{m,s,q,N}^{(\nu)}$ denote the
mean squared reconstruction error for method $m$ in replication $\nu$.
We report
\begin{equation}
\AMSE_{m,s,q,N}
 =\frac1M\sum_{\nu=1}^M L_{m,s,q,N}^{(\nu)},\qquad
L_{m,s,q,N}^{(\nu)}=\frac1N\sum_{i=1}^N
\{\widehat f_{m}^{(\nu)}(t_i)-f_N(t_i)\}^2.
\label{eq:simul-plan-amse}
\end{equation}
The true signal is used solely to \emph{evaluate} reconstruction error,
never to tune a particular realization. Paired replication-level
MSE differences form the natural basis for inferential method
comparisons. To summarize performance across signals of different
difficulty, we define
\begin{equation}
Q_N(m)=\frac1{12}\sum_{s\in\mathcal D}\sum_{q\in\{3,5,7\}}
\frac{\AMSE_{m,s,q,N}}{\AMSE_{{\rm Hard},s,q,N}},
\label{eq:simul-overall-score}
\end{equation}
where $\mathcal D=\{\text{Doppler, HeaviSine, Bumps, Blocks}\}$,
and ``Hard'' is universal hard thresholding at
$\widehat\sigma\sqrt{2\log N}$. Hence $Q_N$ averages 12
\emph{conditionwise normalized} AMSEs with equal weight; it is
neither a ratio of pooled AMSEs nor an oracle selection criterion.
For cross-size summaries we use
$Q_{\rm all}(m)=\{Q_{512}(m)+Q_{1024}(m)+Q_{2048}(m)\}/3$,
without pooling raw MSEs across sample sizes.

\subsection{Experiment I: fixed calibration and neighborhood radius}
\label{sec:simul-fixed}

For each generator---normal, logistic, Laplace, hyperbolic secant,
uniform, and Cauchy---we fix $k=5$ and evaluate
$r\in\{0,1,2,3\}$. Each radius has its own scale
$\lambda_{j,F,r}=a_{F,r,5}(\eta_j)/\widehat\sigma$ from the
Gaussian null-risk integral in Section~\ref{sec:calibration}; thus
differences among radii cannot be attributed to using a common,
uncalibrated scale. Neither SURE nor knowledge of the clean signal
enters this comparison. Table~\ref{tab:e1-normal} reports the
normal-generator AMSEs at $N=1024$ and root SNR $5$.

\begin{table}[tbp]
\centering\small
\caption{Experiment I: AMSE of fixed normal-generator SCOPE ($r=0$) and
N-SCOPE ($r=1,2,3$) for $N=1024$, root SNR $5$, $k=5$, and $M=1000$.
Each radius uses its own Gaussian null-risk calibration. Bold indicates
the smallest AMSE for each signal.}
\label{tab:e1-normal}
\setlength{\tabcolsep}{5pt}
\begin{tabular}{lcccc}
\toprule
Signal & $r=0$ & $r=1$ & $r=2$ & $r=3$ \\
\midrule
Doppler & 0.1610 & 0.1331 & \textbf{0.1278} & 0.1290 \\
HeaviSine & \textbf{0.0696} & 0.0737 & 0.0799 & 0.0852 \\
Bumps & 0.3126 & \textbf{0.2356} & 0.2797 & 0.3218 \\
Blocks & \textbf{0.1240} & 0.2021 & 0.2585 & 0.3021 \\
\bottomrule
\end{tabular}
\end{table}

Neighborhood performance depends markedly on signal morphology.
At $N=1024$ and root SNR $5$, changing from $r=0$ to $r=1$
changes normal-generator AMSE by $-17.4\%$ for Doppler,
$-24.6\%$ for Bumps, $+5.9\%$ for HeaviSine, and
$+63.0\%$ for Blocks. Doppler favors the fixed radius $r=2$,
whereas Bumps favors $r=1$. For Blocks, each positive radius
increases error at isolated coefficients near jumps, notwithstanding
the common null-risk calibration. This outcome reflects the
coefficient-specific structure of N-SCOPE: the central coefficient
is attenuated according to overlapping neighborhood energy, not
by a shrinkage factor shared by an entire block.

Figure~\ref{fig:e1-cdf-gain} compares the effect of $r=1$
across the six generators, averaging the percentage AMSE changes
equally over root SNR $3,5,7$ at $N=1024$.
Figure~\ref{fig:e1-size-snr} examines sample-size sensitivity,
with a separate curve for each root SNR. Each curve compares the
\emph{fixed} normal-generator rule at $r=1$ with its calibrated
$r=0$ counterpart; radii are never selected after observing the
results. The magnitude and even direction of the gain depend on
signal morphology, and particularly for Blocks, on sample size
and strength. Positive-radius neighborhoods are not uniformly
advantageous.

\begin{figure}[t]
\centering
\includegraphics[width=.93\linewidth]{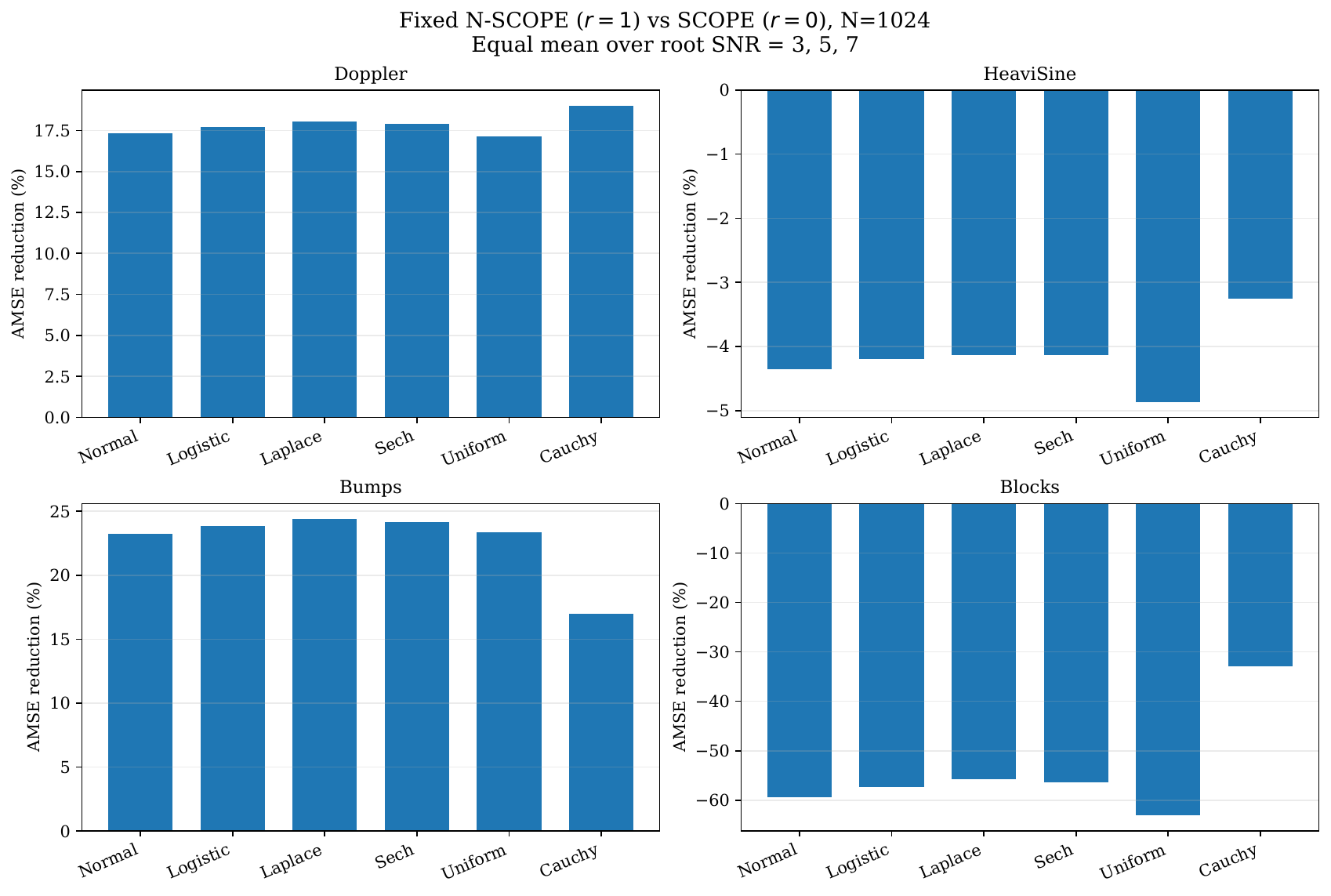}
\caption{Experiment I: percentage reduction in AMSE from fixed $r=0$ to
fixed $r=1$ at $N=1024$, averaged equally over root SNR $3,5,7$, for
each test signal and generator. Positive values favor N-SCOPE. Both
radii are calibrated to the same null-risk target.}
\label{fig:e1-cdf-gain}
\end{figure}

\begin{figure}[t]
\centering
\includegraphics[width=.95\linewidth]{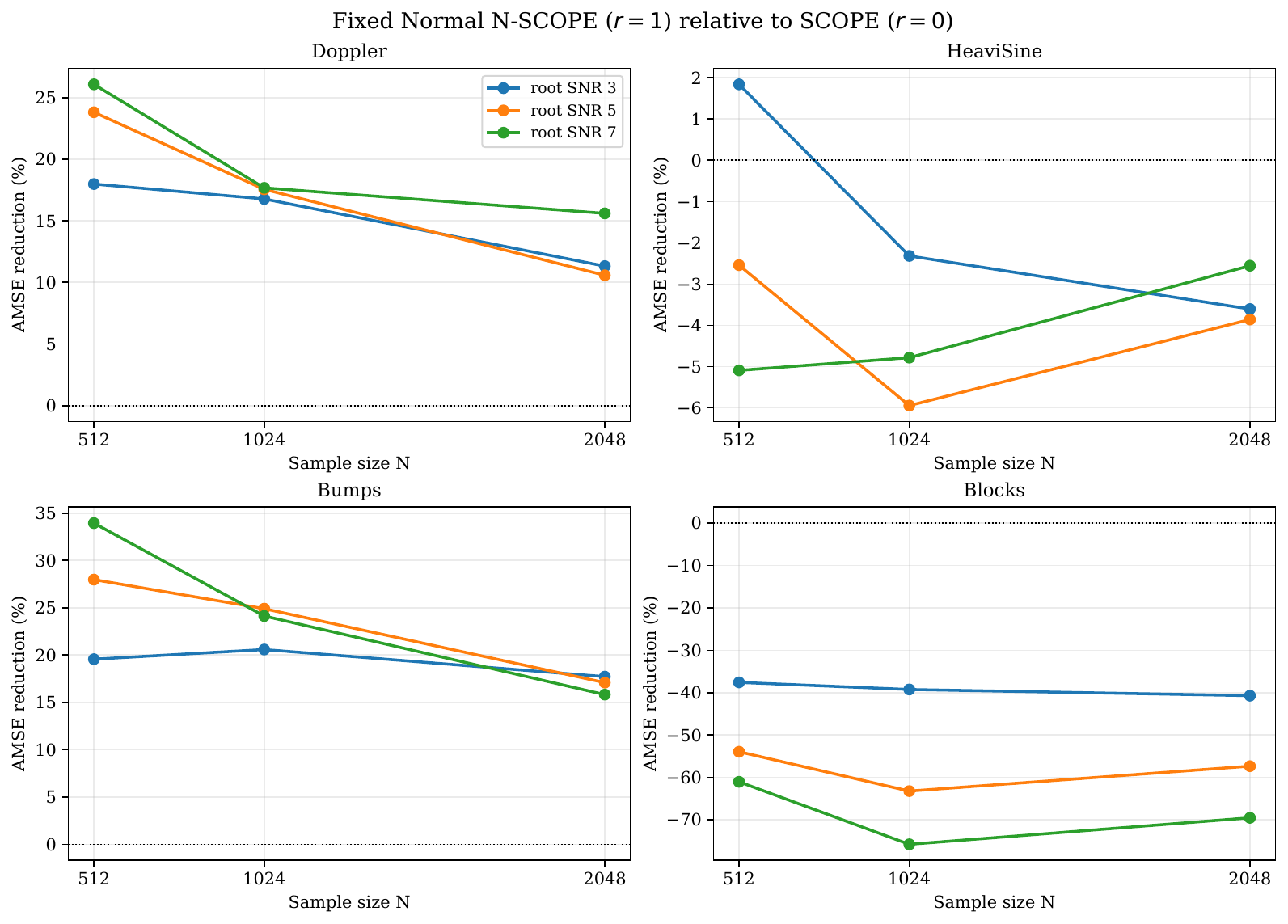}
\caption{Experiment I: sensitivity to sample size and root SNR of the
percentage AMSE reduction of fixed normal-generator N-SCOPE ($r=1$)
relative to coefficientwise SCOPE ($r=0$), for $N=512,1024,2048$. Each
point is based on $M=1000$ paired realizations; positive values
indicate improvement.}
\label{fig:e1-size-snr}
\end{figure}

\subsection{Experiment II: data-adaptive attenuation and neighborhoods}
\label{sec:simul-sure}

For each generator $F$, SCOPE SURE selects $(k,b)$ from 15
candidates with $r=0$, while N-SCOPE SURE selects $(k,r,b)$ from
60 candidates using \eqref{eq:simul-sure-grids}. Candidate scales are
$\lambda_{j,F,r,k,b}=b\,a_{F,r,k}(\eta_j)/\widehat\sigma$.
The plug-in SURE criterion uses the diagonal derivative of the
\emph{central} coefficient estimate, including the contribution
of that coefficient to its local RMS, while treating
$\widehat\sigma$ as fixed. The criterion is minimized jointly
over all shrunk detail levels in each noisy realization, although
the local energies remain position dependent. The generator is
prespecified rather than selected from the data; the true
coefficients enter only AMSE evaluation. As discussed in
Section~\ref{sec:sure-calibration}, the resulting plug-in criterion
is not exactly unbiased after noise estimation and selection.

\begin{table}[tbp]
\centering\small
\caption{Experiment II: paired AMSE of the four normal-generator
procedures for $N=1024$ and root SNR $5$ ($M=1000$). The fixed rules use
$(k,b)=(5,1)$; SURE selects $(k,b)$ with $r=0$, or $(k,r,b)$ with
$r\in\{0,1,2,3\}$. Bold indicates the smallest AMSE for each signal.}
\label{tab:e2-normal}
\setlength{\tabcolsep}{5pt}
\begin{tabular}{lcccc}
\toprule
Signal & Fixed SCOPE & Fixed N-SCOPE & SCOPE SURE & N-SCOPE SURE \\
\midrule
Doppler & 0.1610 & 0.1331 & 0.1640 & \textbf{0.1241} \\
HeaviSine & 0.0696 & 0.0737 & \textbf{0.0577} & 0.0600 \\
Bumps & 0.3126 & 0.2356 & 0.2917 & \textbf{0.2351} \\
Blocks & 0.1240 & 0.2021 & \textbf{0.1211} & \textbf{0.1211} \\
\bottomrule
\end{tabular}
\end{table}

At $N=1024$ and root SNR $5$, normal-generator N-SCOPE SURE
reduces AMSE relative to SCOPE SURE by $24.3\%$ for Doppler
and $19.4\%$ for Bumps, raises it by $4.1\%$ for HeaviSine,
and yields the same reported AMSE for Blocks
(Table~\ref{tab:e2-normal}). Consistent with Experiment~I,
adaptive neighborhoods are especially beneficial for the locally
coherent structures in Doppler and Bumps. For HeaviSine, most of
the improvement arises from adapting attenuation shape and scale;
a radius search adds little and can be slightly unfavorable.
For Blocks, radius zero is available to both SURE procedures,
which coincide at the precision reported here. Nesting guarantees
only that the minimum empirical SURE cannot increase; it does
not guarantee a reduction in true AMSE.

Table~\ref{tab:e2-cdf-ranking} summarizes $Q_N$ by sample size
for the normal and logistic generators. For the normal generator,
N-SCOPE SURE improves $Q_N$ relative to SCOPE SURE by
$10.4\%$, $9.7\%$, and $9.6\%$ at $N=512$, $1024$, and $2048$,
respectively. These values are \emph{relative reductions in
normalized risk}, not percentage-point differences. Section~S5
reports scores for all six generators.

\begin{table}[tbp]
\centering\small
\caption{Experiment II: macro-averaged normalized AMSE $Q_N$, relative
to universal hard thresholding and averaged equally over the four
signals and root SNR $3,5,7$, separately for each sample size. Lower is
better. These are descriptive normalized risks, not pooled raw MSEs.}
\label{tab:e2-cdf-ranking}
\setlength{\tabcolsep}{5pt}
\begin{tabular}{llcccc}
\toprule
$N$ & Generator & Fixed SCOPE & Fixed N-SCOPE & SCOPE SURE & N-SCOPE SURE \\
\midrule
512 & Normal & 0.7105 & 0.7377 & 0.6816 & 0.6106 \\
512 & Logistic & 0.7210 & 0.7467 & 0.6862 & 0.6181 \\
1024 & Normal & 0.7412 & 0.8089 & 0.7005 & 0.6328 \\
1024 & Logistic & 0.7534 & 0.8189 & 0.7090 & 0.6425 \\
2048 & Normal & 0.8025 & 0.8891 & 0.7316 & 0.6614 \\
2048 & Logistic & 0.8121 & 0.8974 & 0.7441 & 0.6727 \\
\bottomrule
\end{tabular}
\end{table}

The selected radius helps explain the observed adaptation.
Figure~\ref{fig:e2-radius-selection} displays selection frequencies
for the normal generator at $N=1024$, averaging equally over the
three root-SNR levels. The preferred radius changes with signal
morphology, from the isolated structures of Blocks to the more
extended oscillatory features of Doppler. Section~S5 reports
selection frequencies at each $N$ for the normal and logistic generators. Because $k$
and $b$ also vary among conditions, the improvement should be
attributed to joint adaptation of attenuation and neighborhood,
not to a universally preferred radius.

\begin{figure}[tbp]
\centering
\includegraphics[width=.90\linewidth]{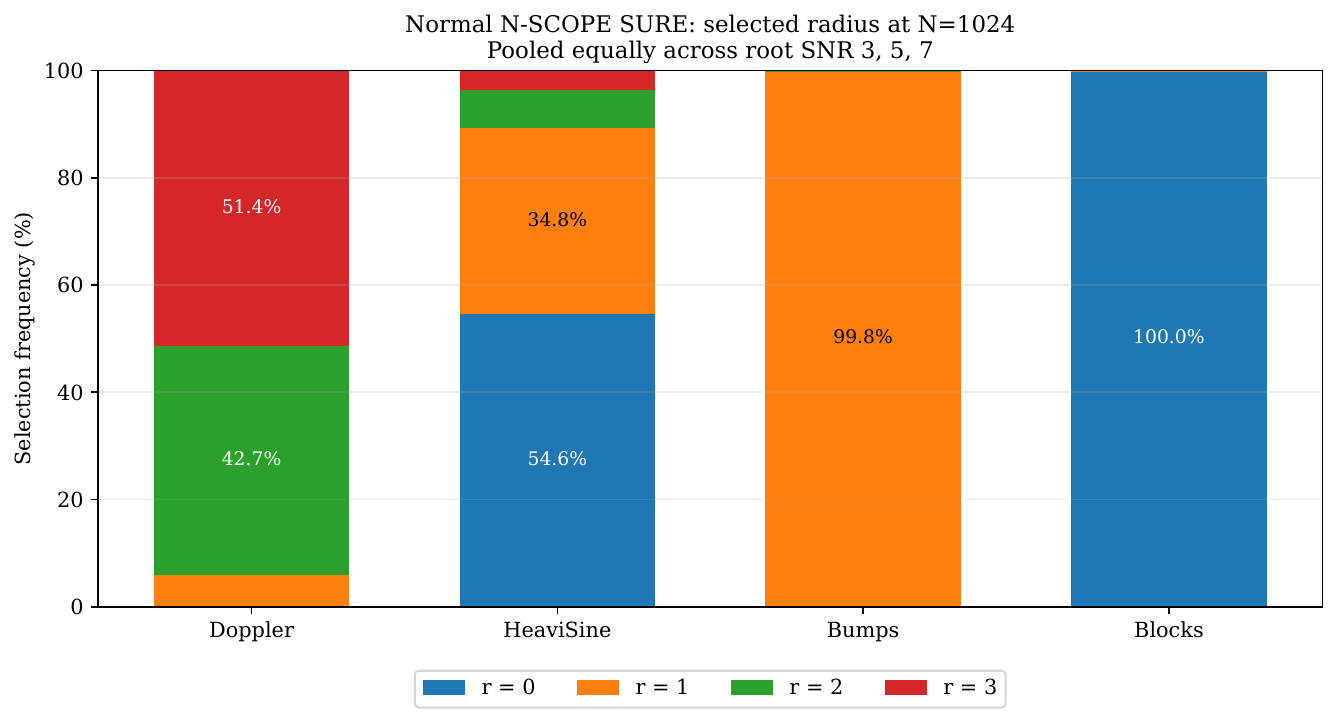}
\caption{Experiment II: selection frequencies of $r=0,1,2,3$ by
normal-generator N-SCOPE SURE at $N=1024$, pooled equally over root SNR
$3,5,7$ ($M=1000$ for each signal and SNR). A radius is selected for
each noisy realization, not for each wavelet coefficient.}
\label{fig:e2-radius-selection}
\end{figure}

\subsection{Experiment III: independently simulated comparator study}
\label{sec:simul-competitors}

We next compare four normal-generator SCOPE procedures specified
before a separate Monte Carlo experiment: fixed SCOPE with
$(k,r,b)=(5,0,1)$, fixed N-SCOPE with $(5,1,1)$, SCOPE SURE
with $r=0$, and N-SCOPE SURE with an adaptive radius. Each uses
the null-risk schedule of Experiments~I and~II. Random streams
are independent across sample sizes, while procedures within each
condition receive identical noise realizations. The inventory
includes 20 published comparator methods spanning empirical Bayes,
smooth, thresholding, neighborhood, and block approaches; one
(BAMS) does not have a reported AMSE in this table. The numerical
results pertain to the implementations used here, not necessarily
to the authors' original software. Several methods were adapted
to one dimension, and published parameter conventions were retained
where sufficiently documented. In particular, the one-dimensional
hidden Markov tree (HMT) implementation is a \emph{provisional proxy};
its results must not be interpreted as those of the published HMT
wavelet denoising method.

\begin{sidewaystable}[p]
\centering
\singlespacing
\singlespacing
\setlength{\tabcolsep}{1.8pt}
\renewcommand{\arraystretch}{0.95}
\fontsize{8}{9}\selectfont
\captionsetup{font=footnotesize,skip=5pt}
\caption{Experiment III: AMSE results for the 26 catalogued
procedures at $N=1024$ and root SNR $3,5,7$ (one comparator has
no reportable entries). Within each SNR group, columns give Doppler
(D), HeaviSine (H), Bumps (Bu), and Blocks (Bl). Each numerical
entry averages $M=1000$ paired replications. Bold indicates the
smallest AMSE in each column among the reported implementations,
including provisional ones.}
\label{tab:e3-all-methods}
\medskip
\begin{tabular}{@{}r p{7.55cm}*{12}{r}@{}}
\toprule
& & \multicolumn{4}{c}{Root SNR $3$} & \multicolumn{4}{c}{Root SNR $5$} & \multicolumn{4}{c}{Root SNR $7$} \\
\cmidrule(lr){3-6}\cmidrule(lr){7-10}\cmidrule(l){11-14}
No. & Procedure and source & D & H & Bu & Bl & D & H & Bu & Bl & D & H & Bu & Bl \\
\midrule
\multicolumn{14}{@{}l}{\emph{SCOPE procedures (normal generator, specified in advance)}}\\
1 & Fixed SCOPE (normal; $r=0$) & 0.1356 & 0.0632 & 0.2826 & 0.1594 & 0.1618 & 0.0699 & 0.3127 & 0.1251 & 0.1804 & 0.0780 & 0.3193 & 0.1054 \\
2 & Fixed N-SCOPE (normal; $r=1$) & 0.1134 & 0.0647 & 0.2253 & 0.2211 & 0.1334 & 0.0740 & 0.2348 & 0.2027 & 0.1493 & 0.0814 & \textbf{0.2428} & 0.1844 \\
3 & SCOPE SURE (normal; $r=0$) & 0.1399 & 0.0503 & 0.2740 & 0.1583 & 0.1648 & \textbf{0.0572} & 0.2911 & 0.1220 & 0.1821 & \textbf{0.0688} & 0.2946 & 0.1037 \\
4 & N-SCOPE SURE (normal) & \textbf{0.1032} & 0.0509 & \textbf{0.2243} & 0.1583 & \textbf{0.1245} & 0.0597 & \textbf{0.2347} & 0.1220 & 0.1384 & 0.0710 & 0.2460 & 0.1037 \\
\addlinespace[2pt]
\multicolumn{14}{@{}l}{\emph{Implementations of published methods (references identify the target methods)}}\\
5 & BAMS \citep{VidakovicRuggeri2001} & -- & -- & -- & -- & -- & -- & -- & -- & -- & -- & -- & -- \\
6 & EBayesThresh posterior median \citep{JohnstoneSilverman2005EBayes} & 0.1441 & 0.0511 & 0.2965 & 0.1662 & 0.1779 & 0.0606 & 0.3168 & 0.1275 & 0.1974 & 0.0724 & 0.3239 & 0.1184 \\
7 & Raised-cosine Bayes \citep{ReinaSousa2026RaisedCosine} & 0.1518 & 0.0503 & 0.3275 & 0.1682 & 0.2036 & 0.0624 & 0.3768 & 0.1276 & 0.2722 & 0.0713 & 0.3980 & 0.1174 \\
8 & Epanechnikov Bayes \citep{BarriosSousa2025Epanechnikov} & 0.1466 & \textbf{0.0467} & 0.3679 & 0.2078 & 0.1872 & 0.0596 & 0.5086 & 0.2028 & 0.2185 & 0.0736 & 0.6711 & 0.2076 \\
9 & Smooth SCAD (stabilized) \citep{KulkarniEtAl2026SmoothSCAD} & 0.1451 & 0.0681 & 0.3121 & 0.1564 & 0.1731 & 0.0667 & 0.3324 & 0.1151 & 0.1885 & 0.0724 & 0.3328 & 0.0966 \\
10 & SureShrink (hybrid) \citep{DonohoJohnstone1995} & 0.2024 & 0.0540 & 0.3614 & 0.3270 & 0.2283 & 0.0695 & 0.4321 & 0.2854 & 0.2900 & 0.0826 & 0.4557 & 0.2157 \\
11 & FDR ($q=0.05$) \citep{AbramovichEtAl2006FDR} & 0.1777 & 0.0997 & 0.3538 & 0.1784 & 0.2059 & 0.0859 & 0.3650 & 0.1304 & 0.2159 & 0.0822 & 0.3547 & 0.1145 \\
12 & Empirical Bayes block (1D) \citep{AbramovichBesbeasSapatinas2002BlockEB} & 0.1480 & 0.0498 & 0.3664 & 0.3482 & 0.2061 & 0.0681 & 0.4142 & 0.3729 & 0.2399 & 0.0904 & 0.4267 & 0.3533 \\
13 & BayesShrink (1D) \citep{ChangYuVetterli2000BayesShrink} & 0.1693 & 0.0488 & 0.3558 & 0.2410 & 0.2347 & 0.0661 & 0.4283 & 0.3167 & 0.2822 & 0.0867 & 0.4999 & 0.4164 \\
14 & NeighCoeff \citep{CaiSilverman2001Neigh} & 0.1326 & 0.0748 & 0.3174 & 0.3401 & 0.1688 & 0.0809 & 0.2974 & 0.3088 & 0.1867 & 0.0864 & 0.2864 & 0.2579 \\
15 & NeighShrink--SURE (1D) \citep{ZhouCheng2008NeighShrinkSURE} & 0.1068 & 0.0592 & 0.2419 & 0.2326 & 0.1253 & 0.0723 & 0.2427 & 0.2184 & \textbf{0.1374} & 0.0835 & 0.2474 & 0.2040 \\
16 & SURE-LET (1D two-term) \citep{BluLuisier2007SURELET} & 0.1559 & 0.0825 & 0.2901 & 0.1781 & 0.1833 & 0.0866 & 0.3179 & 0.1553 & 0.2024 & 0.0949 & 0.3314 & 0.1424 \\
17 &  BiShrink (1D) \citep{SendurSelesnick2002BiShrink} & 0.1096 & 0.0748 & 0.2636 & 0.1921 & 0.1254 & 0.0865 & 0.2828 & 0.2061 & 0.1377 & 0.0949 & 0.2971 & 0.2215 \\
18 & NeighBlock \citep{CaiSilverman2001Neigh} & 0.1218 & 0.0859 & 0.4217 & 0.4680 & 0.1540 & 0.1063 & 0.4021 & 0.4664 & 0.1752 & 0.1139 & 0.3800 & 0.4008 \\
19 & SureBlock (hybrid) \citep{CaiZhou2009SureBlock} & 0.1621 & 0.0513 & 0.3013 & 0.2537 & 0.1822 & 0.0623 & 0.3308 & 0.1713 & 0.2217 & 0.0761 & 0.3336 & 0.1365 \\
20 & BlockJS \citep{Cai1999BlockJS} & 0.1380 & 0.0843 & 0.5599 & 0.4981 & 0.1755 & 0.1177 & 0.5369 & 0.5244 & 0.1967 & 0.1288 & 0.4987 & 0.4534 \\
21 & Jeffreys-prior EB (1D) \citep{FigueiredoNowak2001Jeffreys} & 0.1730 & 0.1132 & 0.2986 & 0.1711 & 0.1916 & 0.1158 & 0.3168 & 0.1543 & 0.2051 & 0.1191 & 0.3230 & 0.1445 \\
22 & Local Bayesian Wiener (1D) \citep{MihcakEtAl1999LocalWiener} & 0.1617 & 0.1295 & 0.3120 & 0.2973 & 0.1732 & 0.1429 & 0.3342 & 0.3185 & 0.1818 & 0.1501 & 0.3472 & 0.3239 \\
23 & Context-based shrinkage (1D) \citep{ChangYuVetterli2000Context} & 0.1557 & 0.0486 & 0.3731 & 0.3403 & 0.2080 & 0.0683 & 0.4436 & 0.4259 & 0.2481 & 0.0904 & 0.4718 & 0.4490 \\
24 & Hidden Markov tree (1D EM) \citep{CrouseNowakBaraniuk1998HMT} & 0.1189 & 0.0607 & 0.2385 & \textbf{0.1012} & 0.1447 & 0.0598 & 0.2413 & \textbf{0.0779} & 0.1639 & 0.0730 & 0.2589 & \textbf{0.0718} \\
\addlinespace[2pt]
\multicolumn{14}{@{}l}{\emph{Universal-threshold references}}\\
25 & Universal hard threshold \citep{DonohoJohnstone1994} & 0.1946 & 0.0989 & 0.4479 & 0.2275 & 0.2376 & 0.0919 & 0.4793 & 0.1460 & 0.2420 & 0.0829 & 0.4875 & 0.1124 \\
26 & Universal soft threshold \citep{DonohoJohnstone1994} & 0.4622 & 0.1202 & 1.2951 & 0.6863 & 0.6525 & 0.2033 & 1.7067 & 0.8250 & 0.7825 & 0.2484 & 1.9976 & 0.8654 \\
\bottomrule
\end{tabular}
\par\smallskip
\begin{minipage}{\linewidth}
\raggedright\fontsize{7.5}{9}\selectfont
BAMS (No.~5) is not reported in this table; see Section~S5.
Citations identify the original methods, not a guarantee of
software equivalence. Methods 6--11, 13--15, 17--18, 20, and 22
implement published rules, in some cases with one-dimensional
adaptations or numerical stabilization; methods 12, 16, 19,
21, 23, and 24 are provisional or modified one-dimensional
implementations and should not be attributed to the original
published algorithms. Section~S5 gives the implementation status of each code.
Root SNR is $\sd(f)/\sigma$, with $\sigma=1$, rather than
power SNR. Methods within each replication receive identical
noise realizations. Section~S5 also provides results for the
other sample sizes.
\end{minipage}
\end{sidewaystable}
Table~\ref{tab:e3-all-methods} presents all available results at
$N=1024$ for each root SNR, rather than highlighting selected
methods or a single noise condition; the BAMS row remains unreported.
The comparisons should be interpreted separately by signal.
For Doppler, normal-generator N-SCOPE SURE has the lowest AMSE
at root SNR $3$ and $5$ ($0.1032$ and $0.1245$), while the
one-dimensional NeighShrink--SURE implementation leads narrowly
at root SNR $7$ ($0.1374$, versus $0.1384$). The margins over this
closest comparator, $0.0036$, $0.0008$, and $-0.0010$ at root SNR
$3$, $5$, and $7$, are small, and interpreting them requires the
Monte Carlo standard errors of the paired replication-level
differences.
For Bumps, N-SCOPE SURE is best at root SNR $3$ and $5$
($0.2243$ and $0.2347$); fixed N-SCOPE with $r=1$ has a
slight advantage at root SNR $7$ ($0.2428$ versus $0.2460$).
For HeaviSine, the Epanechnikov Bayes implementation is best
at root SNR $3$ ($0.0467$), whereas SCOPE SURE is best
at root SNR $5$ and $7$ ($0.0572$ and $0.0688$).
The provisional HMT proxy has the lowest Blocks AMSE in all
three conditions, but this finding cannot be assigned to the
published HMT algorithm. The rankings therefore describe the
tested implementations, not the published methods in general.
Figure~\ref{fig:e3-profile} summarizes signal-specific relative
risks, averaged equally over the three root-SNR values.
Section~S5 gives the normalized-risk rankings of all catalogued
procedures at $N=512,1024,2048$, further condition-specific
tables, and the implementation inventory.

\begin{figure}[t]
\centering
\includegraphics[width=.95\linewidth]{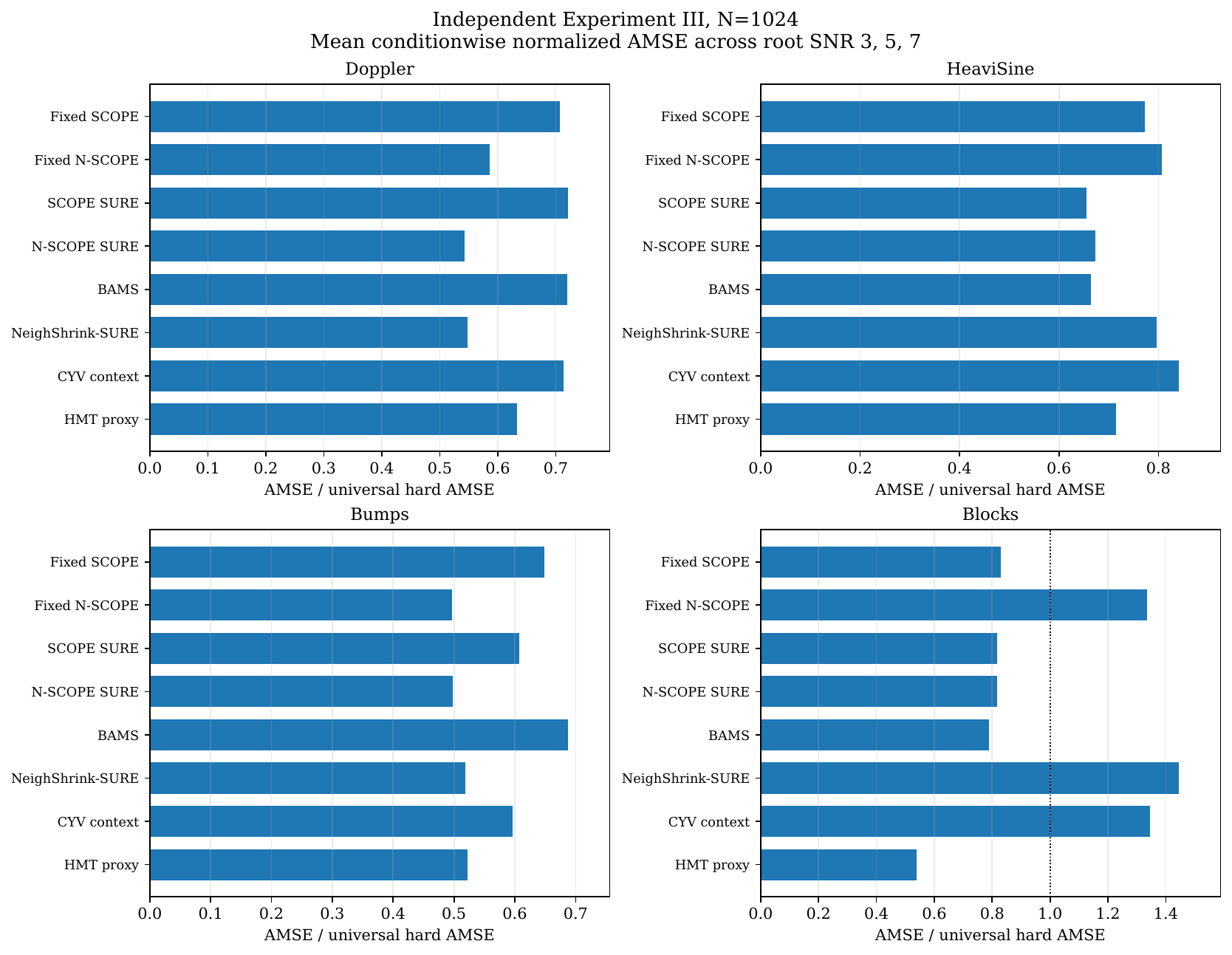}
\caption{Experiment III: mean of the conditionwise AMSE ratios to
universal hard thresholding at $N=1024$, for each signal, averaged
equally over root SNR $3,5,7$. The empirical Bayes posterior median (†) requires implementation audit, while the HMT (‡) is a provisional proxy; see Supplementary Section~S5. All four SCOPE recipes
were specified before these independent replications.}
\label{fig:e3-profile}
\end{figure}

\subsection{Interpretation and limitations}
\label{sec:simul-discussion}

Across the three experiments, neighborhood energy improves recovery
of clustered wavelet features but can be detrimental for isolated
coefficients. Including $r=0$ in the adaptive search largely avoids
the disadvantage of a fixed positive radius on Blocks while
preserving the benefits observed for Doppler and Bumps. This
adaptivity introduces selection variability, however, and does
not imply uniform risk improvement.

Several limitations warrant emphasis. Aligning the half-attenuation
point with the universal threshold is a prespecified calibration
convention, not a derivation of optimal risk scales. The equal-energy
schedule also becomes permissive at $N=2048$, $j=3$. Holding
$\widehat\sigma$ fixed in the SURE derivative ignores its sampling
variability, and SURE unbiasedness does not carry over automatically
to the fully selected estimator. Because each signal is analyzed
with its own wavelet (Haar for Blocks, Daubechies-3 for Bumps, and
Symmlet-4 for Doppler and HeaviSine), the findings describe the tested
signal--wavelet combinations; the design does not separate signal
morphology from wavelet choice. Finally, Experiment~III evaluates
specific implementations on one-dimensional, periodized, orthonormal
wavelet data; adaptations and provisional algorithms limit attribution
to the originating methods. The evidence supports the competitiveness
of N-SCOPE with these implementations, \emph{not} its uniform
superiority over established wavelet denoisers. Presenting results
for all sample sizes and signal conditions is essential to a
balanced comparison.

\section{Discussion and Conclusions}
\label{sec:conclusions}

We have developed SCOPE, a family of CDF-based wavelet shrinkage
rules combining interpretable attenuation profiles with useful
structural properties. The generating distribution, scale, and
exponent determine the shape and location of the shrinkage
transition, ranging from smooth to threshold-like behavior.
Coefficientwise rules preserve sign, contract magnitude, and
retain large coefficients in relative terms, with a connection
to penalized estimation. N-SCOPE extends the framework by using
local wavelet energy while retaining coefficient-specific
attenuation rather than imposing a common blockwise factor.

Both tuning strategies are non-oracle. Fixed SCOPE and N-SCOPE
use resolution-dependent Gaussian null-risk calibration, computed for
each profile and radius, so that different generators and neighborhood
sizes retain the same expected noise before their recovery of signal is
compared; their
adaptive counterparts minimize a plug-in SURE criterion over
attenuation parameters and, for N-SCOPE, the neighborhood radius.
The distinction is between a prespecified noise-based calibration
and greater sensitivity to the observed coefficients. Allowing
$r=0$ within N-SCOPE SURE makes isolated-coefficient protection
available alongside the reinforcement of coherent neighboring
wavelet structure.

For the normal generator, N-SCOPE SURE reduces aggregate normalized
risk by approximately 10\% relative to coefficientwise SCOPE SURE, and
Table~\ref{tab:e2-normal} shows how the gain arises. For Bumps, a fixed
radius $r=1$ already captures nearly all of it (AMSE $0.2356$ versus
$0.2351$ at $N=1024$ and root SNR $5$). For Blocks, the option to return
to $r=0$ recovers coefficientwise performance ($0.1211$) and avoids the
loss from imposing a neighborhood ($0.2021$). Doppler benefits from
joint adaptation of attenuation and radius ($0.1241$, against $0.1331$
and $0.1640$), whereas HeaviSine pays a small price for the larger
search ($0.0600$ versus $0.0577$). Several light-tailed generators
behave similarly, so the procedure does not depend on a single
distributional choice. Among the non-provisional comparator
implementations, N-SCOPE SURE has the lowest aggregate normalized risk;
the provisional HMT proxy scores lower, mainly on Blocks, but cannot
support a claim about the published HMT method. These findings, for
the tested signal--wavelet combinations, demonstrate competitiveness
rather than a uniform ordering of methods.

SCOPE therefore offers a unified approach combining a calibrated
transition location, adjustable attenuation shape, neighborhood
information, and non-oracle parameter selection. In the settings
studied, the choice among the light-tailed generators, and the
compactly supported uniform generator, has relatively little effect,
whereas the heavy-tailed Cauchy generator performs markedly less
favorably (Section~S5).
Promising extensions include resolution-dependent radius selection,
redundant wavelet transforms, multivariate signals, and inverse
problems. More generally, centered distribution functions provide
a simple and interpretable basis for designing competitive
shrinkage estimators.

\section*{Code availability}

The MATLAB code for the SCOPE simulation studies (Experiments I--III),
Python notebooks, and R implementations are available at
\url{https://github.com/Vijini95/SCOPE.git}. The repository also includes
the supplementary material and archived simulation summaries.

\section*{Acknowledgments}

B.~Vidakovic acknowledges the partial support of the H.~O.~Hartley
Chair Foundation and NSF Award 2515246 at Texas A\&M University.

\bibliographystyle{apalike}
\bibliography{references}

\end{document}